\documentclass[12pt]{article}
\usepackage[utf8]{inputenc}
\usepackage[english]{babel}
\usepackage{natbib}

\usepackage{amsthm}
\usepackage{amssymb}
\usepackage{amsmath}
\usepackage{mathrsfs}
\usepackage{url}
\usepackage{latexsym}
\usepackage{wasysym}
\usepackage{newlfont}
\usepackage{leqno}
\usepackage{xspace}
\usepackage{blkarray}
\usepackage{diagbox}
\usepackage{bbm}
\usepackage[prependcaption]{todonotes}

\usepackage{tikz}
\usetikzlibrary{arrows}

\usepackage[capitalize]{cleveref}
\newtheorem{theorem}{Theorem}[section]

\newtheorem{lemma}[theorem]{Lemma}
\newtheorem{proposition}[theorem]{Proposition}
\newtheorem{corollary}[theorem]{Corollary}
\newtheorem{definition}[theorem]{Definition}
\newtheorem{remark}[theorem]{Remark}

\newtheorem{example}[theorem]{Example}
\newtheorem{algorithm}{Algorithm}[section]

\crefname{equation}{eq.}{equations}
\crefname{definition}{def.}{definitions}
\crefname{section}{sec.}{sections}
\crefname{lemma}{lem.}{lemmata}
\crefname{observation}{obs.}{observations}
\crefname{corollary}{cor.}{corollaries}
\crefname{proposition}{prop.}{propositions}
\crefname{remark}{rem.}{remarks}
\crefname{enumi}{step}{steps}
\crefname{theorem}{thm.}{theorems}
\crefname{algorithm}{alg.}{algorithms}
\crefname{figure}{fig.}{figures}

\usepackage[ruled]{algorithm}
\usepackage{algorithmic}

\DeclareFontFamily{U}{mathx}{\hyphenchar\font45}
\DeclareFontShape{U}{mathx}{m}{n}{
      <5> <6> <7> <8> <9> <10>
      <10.95> <12> <14.4> <17.28> <20.74> <24.88>
      mathx10
      }{}
\DeclareSymbolFont{mathx}{U}{mathx}{m}{n}
\DeclareFontSubstitution{U}{mathx}{m}{n}
\DeclareMathAccent{\widecheck}{0}{mathx}{"71}
\DeclareMathAccent{\wideparen}{0}{mathx}{"75}

  \def\bm{\boldsymbol}

\def\Nz{\mathbb{N}}
\def\N{\mathbb{N}}
\def\Z{\mathbb{Z}}
\def\Q{\mathbb{Q}}

\def\K{\mathbb{K}}
\def\R{\mathbb{R}}
\def\Pr{\mathbb{P}}
\def\P{\mathcal{P}}

\def\b{\mathfrak{b}}

\def\rows{\mathtt{Rows}}

\def\lm{{\mathtt{LM}}}
\def\LM{\lm}

\def\sm{{\prec}}
\def\smh{{\prec_h}}

\def\groebner{Gr\"{o}bner\xspace}
\def\grobner{\groebner}

\def\Modu{\mathrm{M}}

\def\M{\mathcal{M}}

\def\spDeg{\delta}
\def\degHom{{\mathrm{deg}}}
\def\degAff{\delta^A}

\def\dehom{\chi}
\def\hom{\dehom^{-1}}

\newcommand{\mon}[1]{\bm{X}^{#1}}
\newcommand{\monhom}[2]{\bm{X}^{({#1}, {#2})}}
\newcommand{\divs}[2]{{#1} || {#2}}

\usepackage{multirow}

\PassOptionsToPackage{usenames}{color}
\usepackage{tikz}
\usetikzlibrary{matrix}

\usepackage{stmaryrd}

\newcommand\possiblebreak{\ifhmode\unskip\space\hfil\penalty0\hfilneg\fi}

\newcommand{\Supp}{\texttt{sp}}

\makeatletter
\newif\if@borderstar
\def\bordermatrix{\@ifnextchar*{%
  \@borderstartrue\@bordermatrix@i}{\@borderstarfalse\@bordermatrix@i*}%
}
\def\@bordermatrix@i*{\@ifnextchar[{%
  \@bordermatrix@ii}{\@bordermatrix@ii[()]}
}
\def\@bordermatrix@ii[#1]#2{%
  \begingroup
    \m@th\@tempdima8.75\p@\setbox\z@\vbox{%
      \def\cr{\crcr\noalign{\kern 2\p@\global\let\cr\endline }}%
      \ialign {$##$\hfil\kern 2\p@\kern\@tempdima & \thinspace %
      \hfil $##$\hfil && \quad\hfil $##$\hfil\crcr\omit\strut %
      \hfil\crcr\noalign{\kern -\baselineskip}#2\crcr\omit %
      \strut\cr}}%
    \setbox\tw@\vbox{\unvcopy\z@\global\setbox\@ne\lastbox}%
    \setbox\tw@\hbox{\unhbox\@ne\unskip\global\setbox\@ne\lastbox}%
    \setbox\tw@\hbox{%
      $\kern\wd\@ne\kern -\@tempdima\left\@firstoftwo#1%
        \if@borderstar\kern2pt\else\kern -\wd\@ne\fi%
      \global\setbox\@ne\vbox{\box\@ne\if@borderstar\else\kern 2\p@\fi}%
      \vcenter{\if@borderstar\else\kern -\ht\@ne\fi%
        \unvbox\z@\kern-\if@borderstar2\fi\baseline skip}
        \if@borderstar\kern-2\@tempdima\kern2\p@\else\,\fi\right\@secondoftwo#1 $%
    }\null \;\vbox{\kern\ht\@ne\box\tw@}%
  \endgroup
}
\makeatother

\usepackage{soul}

\usepackage{amssymb,amsfonts,amsmath,xspace}

\begin{document}

\title{Towards Mixed Gr\"{o}bner Basis Algorithms: \\ the Multihomogeneous and Sparse Case}

\author{Mat\'{i}as R. Bender, Jean-Charles Faug\`ere, and Elias Tsigaridas}

\date{
Sorbonne Universit\'e, \textsc{CNRS}, \textsc{INRIA},\linebreak
    Laboratoire d'Informatique de Paris~6, \textsc{LIP6},
    \'Equipe \textsc{PolSys}, \linebreak
    4 place Jussieu, F-75005, Paris, France
    \linebreak
    \\
    March 2018\footnote{This work was originally published in ISSAC ’18, July 16–19, 2018, New York, USA}
  }


\maketitle

\vspace{-20px}

\begin{abstract}

  One of the biggest open problems in computational algebra is the
  design of efficient algorithms for Gr\"{o}bner basis computations
  that take into account the sparsity of the input polynomials.
  We can perform such computations in the case of unmixed polynomial
  systems, that is systems with polynomials having the same support, using
  the approach of 
  Faug\`ere, Spaenlehauer, and Svartz [ISSAC'14].
%
  We present two algorithms for sparse Gr\"{o}bner bases computations
  for mixed systems.
  The first one computes with mixed sparse systems and exploits the
  supports of the polynomials.
  Under regularity assumptions, it performs no reductions
  to zero.
  For mixed, square, and 0-dimensional multihomogeneous polynomial
  systems, we present a dedicated, and potentially more efficient,
  algorithm that exploits different algebraic properties that performs
  no reduction to zero.  We give an explicit bound for the
  maximal degree appearing in the computations.
\end{abstract}



\section{Introduction}


\grobner bases are  in the heart of many algebraic
algorithms. One of the most important applications is to solve
0-dimensional polynomial systems.  A common strategy is, first to
compute a \grobner basis in some order, usually degree
lexicographic, deduce from it multiplication maps in the
corresponding quotient ring, and finally recover the lexicographic
order using FGLM \cite{faugere1993efficient}.


  Toric geometry~\cite{cox2011toric} studies the
  geometric and algebraic properties of varieties given by the image
  of monomial maps and systems of \emph{sparse} polynomial equations;
  that is systems with polynomials
  having monomials from a restrictive set.  
  Sparse resultant~\cite{gelfand2008discriminants}, that generalizes
  the classical multivariate resultant, extends these ideas in
  (sparse) elimination theory.  There are a lot of algorithms to
  compute the sparse resultant and to solve sparse systems, for
  example
  see~\cite{sturmfels1994newton,emiris1996complexity,d2002macaulay}.
  For the related problem of fewnomial systems see
  \cite{bihan2011fewnomial}.  Numerical continuation methods can also
  benefit from sparsity~\cite{li1997numerical}, as well as other
  symbolic algorithms~\cite{giusti2001grobner,herrero2013affine}.


  Recently \citet{faugere2014sparse} introduced the first algorithm to
  solve \emph{unmixed} sparse systems, that is systems of sparse 
  polynomials that have the same monomials, using \grobner basis that
  exploits sparsity.
  Their idea is to consider the polytopal algebra associated to the
  supports of the input polynomials.
  Roughly speaking, the polytopal algebra is like the
  standard polynomial algebra, where the variables are the monomials
  in the supports of the input polynomials.
  They compute a \grobner basis
  of the ideal generated by the polynomials, in the polytopal 
  algebra,
  by introducing a matrix F5-like algorithm
  \citep{faugere2002F5,eder2014survey}.
  They homogenize the polynomials and compute a
  \grobner basis degree by degree. By
  dehomogenizing the computed basis, they recover a \grobner basis
  of the original ideal. In the 0-dimensional case, they apply a
  FGLM-like algorithm~\cite{faugere1993efficient} to obtain a lexicographical 
  \grobner basis.
  If the homogenized polynomials form a regular sequence over the
  polytopal algebra, then the algorithm performs no reductions to
  zero. When the system is also 0-dimensional, they bound the
  complexity using the Castelnuovo-Mumford regularity.
  In this case, taking advantage of the sparsity led to large
  speed-ups.
  Hence, our goal is to extend~\cite{faugere2014sparse} to
  \emph{mixed} sparse polynomial systems, i.e. systems where the
  polynomials {\em do not} have necessarily the same monomials.

  
  The Castelnuovo-Mumford regularity is a fundamental invariant in
  algebraic geometry,
  related to the maximal
  degrees appearing in the minimal resolutions and the vanishing of
  the local cohomology. It is related to the complexity
  of computing \grobner basis~\cite{bayer1993can,chardin2007some}.
  The extension of this regularity in the context of toric varieties
  is known as multigraded Castelnuovo-Mumford regularity~\cite{maclagan2003uniform,maclagan_multigraded_2004,botbol_castelnuovo_2017}.


  The multihomogeneous systems form an important subclass of mixed
  sparse systems as they are ubiquitous in applications.  Their
  properties are well understood, for example, the degree (number of
  solutions) of the system~\cite{van1978varieties}, the arithmetic
  Nullstellens\"atze~\cite{d2011heights}, and the (multigraded)
  Castelnuovo-Mumford
  regularity~\cite{ha_regularity_2004,awane2005formes,sidman2006multigraded,botbol_implicitization_2011,botbol_castelnuovo_2017}.
  We can solve these systems using general purpose algorithms based on
  resultants~\cite{emiris1996complexity} and in some cases benefit
  from the existence of determinantal formulas~\cite{sturmfels1994multigraded,weyman1994multigraded},
  or we can use homotopy methods~\cite{hauenstein2015multiprojective,el2017bit}.
  For unmixed bilinear systems, we  compute a \grobner basis~\cite{faugere2011grobner} with no reductions to zero.  Using
  determinantal formulas we can solve mixed bilinear systems with two
  supports using
  eigenvalues/eigenvectors~\cite{bender_2bilinear_2018}.  In the
  unmixed case, \cite{faugere2014sparse} presents bounds for the complexity of
  computing a sparse \grobner basis.
  Our goal is to present a potentially more efficient  algorithm and bounds for
  square mixed multihomogeneous~systems.

\paragraph{Our contribution}
\label{sec:our-contribution}

We present two algorithms to solve $0$-di\-men\-sio\-nal mixed sparse
polynomial systems based on \grobner basis computations. Both of
them, under  assumptions, compute with no reductions to
zero, thus they avoid useless computations.

The first algorithm (\Cref{alg:matrixF5}) takes as input a mixed
sparse system and computes a \emph{sparse \grobner basis}
(\Cref{def:sparse-grobner-basis}).  This is a basis for the
corresponding ideal over a polytopal algebra and has similar
properties to the usual \grobner basis. Using this basis, we 
compute normal forms by a modified division algorithm
(\Cref{thm:divisionConverges}).
The orders for the monomials that we consider take into account the
supports of the polynomials and they \emph{are not} necessarily
monomial orders (\Cref{sec:monomial-orders-well}).  We prove that for
any of these orders and any ideal there is a finite sparse \grobner
basis (\Cref{thm:finiteHomogeneousSGB,thm:finiteAffineSGB}) that we
compute with a matrix F5-like algorithm, that we call $\texttt{M}^2$.
Moreover, we introduce a \emph{sparse F5 criterion} to avoid useless
computations. Under regularity assumptions, we avoid every reduction
to zero (\Cref{thm:sparse-f5}).
When the ideal is 0-dimensional, we can use a sparse \grobner
basis to compute a \grobner basis for unmixed systems introduced in 
\cite{faugere2014sparse} using FGLM.

Our second algorithm, $\texttt{M}_3\texttt{H}$, takes as input a 0-dimensional
square multihomogeneous mixed system, that has no solutions at
infinity. It outputs a monomial basis
and the multiplication map of every affine variable.
Both lie  in the quotient ring of the
dehomogenization of the (input) ideal.
Using the multigraded Castelnuovo-Mumford regularity, we present an
algorithm (\Cref{alg:matrixMultihom}) that avoids all reductions to
zero (\Cref{thm:noRedToZeroMultihom}).
Over $\Pr^{n_1} \times \cdots \times \Pr^{n_r}$, if the input
polynomials have multidegrees
$\bm{d}_1,\dots,\bm{d}_{(n_1 + \dots + n_r)} \in \N^r$,
then
the dimension
of the biggest matrix appearing in the computations is the
number of monomials of multidegree
$  \sum_{i=1}^{ n_1 + \dots + n_r} \bm{d}_i + (1,\dots,1) - (n_1,\dots,n_r) $.
This bounds the maximal degree of the polynomials appearing in the computations and generalizes the classical Macaulay bound \cite{lazard1983grobner},
which we recover for $r=1$.
Using the  multiplication matrices, we can recover the usual \grobner
basis for the dehomogenized ideal via FGLM.

\section{Preliminaries}

Let $\K$ be a field of characteristic 0,
$\bm{y} := (y_0,\dots,y_m)$, and
$\K[\bm{y}] := \K[y_0,\dots,y_m]$.
For $\alpha \in \Nz^{m+1}$, let
$\bm{y}^\alpha := \prod_{i = 0}^m y_i^{\alpha_i}$.
Let $\bar{0} := (0 \dots 0)$.

  \subsection{Semigroup Algebra}
  \label{sec:semigroup-algebra}
 
  An affine semigroup $S$ is a finitely-generated additive
  subsemigroup of $\Z^n$, for some $n \in \N$, such that it contains
  ${0} \in \Z^n$. The semigroup algebra $\K[S]$ is the
  $\K$-algebra generated by $\{\mon{s}, s \in S\}$, where
  $\mon{s} \cdot \mon{t} = \mon{s + t}$. The set of monomials of
  $\K[S]$ is $\{ \mon{s} , s \in S\}$.

  Let $\{a_0,a_1,\dots,a_m\}$ be a set of generators of
  $S \subset \Z^{n}$. Let $e_0\dots e_{m}$ be the canonical basis
  of $\Z^{m+1}$. Consider the homomorphism $\rho : \Z^{m+1} \rightarrow S$ that
  sends $e_i$ to $a_i$, for $0 \leq i \leq m$. Then,
  $\K[S]$ is isomorphic to the quotient ring
  $\K[\bm{y}] / T$, where $T$ is the lattice ideal
  $T := \langle \bm{y}^{u} - \bm{y}^{v} | u,v \in \Nz^{m+1}, \rho(u-v)
  = 0 \rangle$ \cite[Thm~7.3]{miller2004combinatorial}. Moreover, the
  ideal $T$ is prime and $\K[S]$ is an integral
  domain~\cite[Thm~7.4]{miller2004combinatorial}.

  An affine semigroup $S$ is pointed if it does not contain non-zero
  invertible elements, that is for all
  $s, t \in S \setminus \{\bar{0}\}$, $s + t \neq 0$
  \cite[Def~7.8]{miller2004combinatorial}. As in~\cite{faugere2014sparse}, we consider only pointed affine
  semigroups.

  Let $M_1,\dots,M_k \subset \R^n$ be polytopes containing
  $0$. We consider two different semigroups associated to them.
  First, we consider the affine semigroup $(S_{M_1,\dots,M_k},`+`)$
  generated by the elements in $\cup_{i=1}^k (M_i \cap \Z^n)$ with the addition
  over $\Z^n$. Second, we consider the affine semigroup
  $(S_{M_1,\dots,M_k}^h,`+`)$, generated by the elements in
  $\cup_{i=1}^k \{(s,e_i) : s \in M_i \cap \Z^n\}$, 
  with the addition over $\Z^{n+k}$, where
  $e_1,\dots,e_k$ is the standard basis of $\R^k$.

  \subsection{Sparse degree and homogenization}
  \label{sec:sparse-degree-and-homogenization}
  
  Given a monomial $\monhom{s}{d} \in \K[S_{M_1,\dots,M_k}^h]$, we
  define its \emph{degree} as $\degHom(\monhom{s}{d}) := d \in \Nz^k$.
  With this grading, the semigroup algebra $\K[S_{M_1,\dots,M_k}^h]$
  is multigraded by $\Nz^k$ and generated, as a $\K$-algebra, by the
  elements of degrees $e_1,\dots,e_k$, so it is multihomogeneous.
  For each $d \in \Nz^k$, let $\K[S_{M_1,\dots,M_k}^h]_d$ be the
  vector space of the multihomogeneous polynomials in
  $\K[S_{M_1,\dots,M_k}^h]$ of degree $d \in \N^k$.

  We define the dehomogenization of $\monhom{s}{d}$ as the epimorphism
  that takes $\monhom{s}{d} \in \K[S_{M_1 \dots M_k}^h]$ to
  $\dehom(\monhom{s}{d}) = \mon{s} \in \K[S_{M_1 \dots M_k}]$.
  For
  an ideal $I^h$, $\dehom(I^h)$ means that we apply $\dehom$ to the
  elements of $I^h$.
  
  \begin{remark} \label{rmk:biggerDegSameDehom}
    For an ideal $I^h$, for every $f \in I^h \cap \K[S_{M_1,\dots,M_k}^h]_d$, and $D \geq d$,
    component-wise, there is $f' \in I^h \cap \K[S_{M_1,\dots,M_k}^h]_D$ such
    that $\dehom(f) = \dehom(f') \in \dehom(I^h)$.
  \end{remark}
  
  When we work only with one polytope $M$, that is $k=1$, we define
  the \emph{affine degree} of $\mon{s} \in \K[S_M]$,
  $\degAff(\mon{s})$, as the smallest $d \in \N$ such that
  $\monhom{s}{d} \in \K[S_M^h]$.
  We extend this definition to the affine polynomials in $\K[S_M]$ as
  the maximal affine degree of each monomial. That is, for
  $f := \sum_{s \in S_M} c_s \mon{s} \in \K[S_M]$, the affine degree of
  $f$ is $\degAff(f) := \max_{s\in S_M}(\degAff(\mon{s}) : c_s \neq 0)$.
  Let $\K[S_M]_{\leq d}$ be the set of all polynomials in $\K[S_M]$ of
  degree at most $d$.
  The map $\hom : \K[S_M] \rightarrow \K[S_M^h]$ defines the
  homogenization of $f := \sum_{s \in S_M} c_s \mon{s} \in \K[S_M]$,
  where
  $\hom(f) := \sum_{s \in S_M} c_s \monhom{s}{\degAff(f)} \in
  \K[S_M^h]$. Note that this map is not a homomorphism.  For an ideal
  $I$, $\hom(I)$ is the homogeneous ideal generated by applying $\hom$
  to every element of $I$.
  
  Finally, given a polynomial $f \in \K[S_M^h]$ we define its
  \emph{sparse degree} as $\spDeg(f) := \degAff(\dehom(f))$.
  Note that, the degree is always bigger or equal to the sparse
  degree.
  Even though we use the name sparse degree, it does not give a
  graded structure to the $\K$-algebra $\K[S_M^h]$.

  \subsection{Mixed systems and Regularity}
  \label{sec:mixed-regularity}
  
  Consider polytopes $M_1,\dots,M_k$ and a polynomial system
  $(f_1 \dots f_k)$ such that
  $f_i \in \K[S_{M_1,\dots,M_k}^h]_{e_i}$.
%
  We say the system is \emph{regular} if $f_1,\dots,f_k$ form a
  regular sequence over $\K[S_{M_1,\dots,M_k}^h]$.
  Similarly, $( \dehom(f_1)\dots\dehom(f_k) )$, that is the
  dehomogenization of $( f_1,\dots,f_k )$, is~\emph{re\-gular} if
  $(\dehom(f_1)\dots\dehom(f_k))$ form a
  regular sequence
  over $\K[S_{M_1 \dots M_k}]$.

  When all the polytopes are the same these definitions \emph{match}
  the definition of regularity for unmixed systems~\cite{faugere2014sparse}. When every polytope is a $n$-simplex,
  these definitions are related to the standard definition of
  regularity~\cite[Chp.~17]{eisenbud2013commutative}.

  Like in the (standard) homogeneous case, the order of the polynomials
  does not affect the regularity of the system $(f_1.\dots,f_k)$. In
  addition, the dehomogenization preserves the regularity property.
  
  \begin{lemma}
    Consider $f_i \in \K[S_{M_1,\dots,M_k}^h]_{e_i}$ and $\sigma$ a
    permutation of $\{1,\dots,k\}$. If $f_1,\dots,f_k$ is a regular
    sequence over $\K[S_{M_1,\dots,M_k}^h]$, then
    $(\dehom(f_{\sigma_1}), \dots, \dehom(f_{\sigma_k}))$ is a regular sequence
    over $\K[S_{M_1,\dots,M_k}]$.
  \end{lemma}

  \begin{proof}
%
    If $f_1,\dots,f_k$ is a regular sequence, then any permutation of
    them it is regular~\cite[\S 9,
    Cor. 2]{bourbaki2007algebre}. Hence, we just have to prove that
    $\dehom(f_1),\dots,\dehom(f_k)$ is a regular sequence.
    For $w \leq k$, consider a polynomial
    $\bar{g}_w \in \K[S_{M_1,\dots,M_k}]$ such that
    $\bar{g}_w \cdot \dehom(f_w) \in \langle \dehom(f_1), \dots,
    \dehom(f_{w-1}) \rangle$. Then, there are polynomials
    $\bar{g}_1,\dots,\bar{g}_{w-1} \in \K[S_{M_1,\dots,M_k}]$ such
    that $\sum_{i=1}^w \bar{g}_i \dehom(f_i) = 0$. As $\dehom$ is an
    epimorphism, for each $\bar{g}_i$, there is
    $g_i \in \K[S_{M_1,\dots,M_k}^h]$ multihomogeneous such that
    $\dehom({g_i}) = \bar{g}_i$.
    Consider a vector $D$, such that $\forall i,j$,
    $D - \deg(f_i) \geq \deg(g_j)$. Then, by
    \Cref{rmk:biggerDegSameDehom}, there are
    multihomogeneous polynomials
    $g_i' \in \K[S_{M_1,\dots,M_k}^h]_{D - \degHom(f_i)}$, such that
    $\dehom({g_i'}) = \bar{g}_i$. Note that, $\dehom$ restricted to \linebreak
    $\K[S_{M_1,\dots,M_k}^h]_D$ is injective. Hence
    $\dehom(\sum_{i=1}^w {g_i'} f_i) =\sum_{i=1}^w \bar{g}_i
    \dehom(f_i) = 0$ implies $\sum_{i=1}^w {g_i'} f_i = 0$. As
    $f_1,\dots,f_w$ is a regular sequence,
    $\small {g_w'} \in \langle f_1, \dots, f_{w-1} \rangle$ and
    $\small \bar{g}_w \in \langle \dehom(f_1), \dots, \dehom(f_{w-1})
    \rangle$.
  \end{proof}

  The proof of existence of regular systems is beyond the scope of
  this paper.  Nevertheless, we can report that we have performed
  several experiments with many different sparse mixed systems, taking
  generic coefficients, and all them  were regular.

  \subsection{Orders for  Monomials}
  \label{sec:monomial-orders-well}
  As in the standard case, a monomial order $<$ for $\K[S]$ is a
  well-order compatible with the multiplication on $\K[S]$, that is
  $\forall s \in S, s \neq 0 \implies \mon{0} < \mon{s}$ and
  $\forall s,r, t \in S, \mon{s} < \mon{r} \implies \mon{s+t} <
  \mon{r+t}$.  These orders exist on $\K[S]$ if and only if $S$ is
  pointed, \cite[Def~3.1]{faugere2014sparse}.
    
  Given any well-order $<$ for $\K[S_M]$, we can extend it to a
  well-order $<_{h}$, the grading of $<$, for $\K[S_M^h]$ as follows:
  \begin{equation} \label{eq:graduateOrder}
    \monhom{s}{d} < \monhom{r}{d'} \iff
    \begin{cases}
      d < d' & \\ d = d' \land \mon{s} < \mon{r}
    \end{cases}
  \end{equation}
  If $<$ is a monomial order, then $<_{h}$ is a monomial
  order too.
  
  Given an ideal $I \subset \K[S_M]$, a common issue is to study the
  vector space $I \cap \K[S_M]_{\leq d}$, i.e. the elements of $I$ of
  degree smaller or equal to $d$. This information allow us, for
  example, to compute the Hilbert Series of the affine ideal.
%
  It is also important for computational reasons. For example, to
  maintain the invariants in the signature-based \groebner basis
  algorithms, as the F5 algorithm~\cite{faugere2002F5,eder2014survey}.

  In our setting, to compute a basis of $I \cap \K[S_M]_{\leq d}$, we
  have to work with an order for the monomials in $\K[S_M]$ that takes
  into account the sparse degree.
  This order, $\prec$, is such that for any
  $\mon{s}, \mon{r} \in \K[S_M]$,
  $\degAff(\mon{s}) < \degAff(\mon{r}) \implies \mon{s} \prec
  \mon{r}$. Unfortunately, for most of the polytopal algebras
  $\K[S_M]$, \emph{there is no monomial order} with this
  property. Therefore, we are forced to work with well-orders that are
  not monomial orders.
  \begin{example}
    Consider the semigroup generated by \linebreak
    $M := \{[0,0], [1,0], [0,1], [1,1]\} \subset \N^2$. Consider a
    monomial order $<$ for $\K[S_M]$. Without loss of generality,
    assume $\mon{[1,0]} < \mon{[0,1]}$. Then,
    $\mon{[2,0]} < \mon{[1,1]} < \mon{[0,2]}$. But,
    $\degAff(\mon{[2,0]}) = 2$ and $\degAff(\mon{[1,1]}) = 1$. So,
    no monomial order on $\K[S_M]$ takes into account the sparse degree.
  \end{example}
  
  Given a monomial order $<_M$ for $\K[S_M]$, we define the
  \emph{sparse order} $\prec$ for $\K[S_M]$ as follows.
  \begin{equation} \label{eq:ordFromMonomial}
    \mon{s} \prec \mon{r} \iff
    \begin{cases}
      \degAff(\mon{s}) < \degAff(\mon{r}) \\
      \degAff(\mon{s}) = \degAff(\mon{r}) \land \mon{s} <_M \mon{r}
    \end{cases}
  \end{equation}

  Let $\prec_{h}$ be the grading of the sparse order of
  $\K[S_M^h]$ (Eq.~\ref{eq:graduateOrder}). We call this order the
  \emph{graded sparse order}.

  \begin{remark} \label{thm:homDehomOrder}
    By definition, these two orders
    are the same for monomials of the same degree. That is,
    $$
    \forall \monhom{s}{d}, \monhom{r}{d} \in \K[S_M^h],
    \, \monhom{s}{d} \prec_{h} \monhom{r}{d} \iff \mon{s} \prec
    \mon{r} \enspace.$$
  \end{remark}
  Usually, this order is not compatible with the
  multiplication. But,
  \begin{lemma} \label{thm:compatMultAff}
    If $\mon{s} \prec \mon{t}$ and
    $\degAff(\mon{r}) + \degAff(\mon{t}) = \degAff(\mon{t} \cdot
    \mon{r})$, then
    $\mon{s} \cdot \mon{r} \prec \mon{t} \cdot \mon{r}$.
  \end{lemma}
  \begin{proof}
    Note that $\degAff$ satisfies the triangular inequality,
    $\degAff(\mon{s + r})$ $\leq \degAff(\mon{s}) +
    \degAff(\mon{r})$.
    As $\mon{s} \prec \mon{t}$,
    $\degAff(\mon{s}) \leq \degAff(\mon{t})$.
    By assumption,
    $\degAff(\mon{t}) + \degAff(\mon{r}) = \degAff(\mon{t + r})$.
    So,
    $\degAff(\mon{s + r}) \leq \degAff(\mon{s}) + \degAff(\mon{r}) \leq 
    \degAff(\mon{t}) + \degAff(\mon{r}) \leq \degAff(\mon{t + r})$.
    Hence, either $\degAff(\mon{s + r}) < \degAff(\mon{t + r})$ or
    the sparse degree is the same. In the second case, we conclude
    $\degAff(\mon{s}) = \degAff(\mon{t})$, and so
    $\mon{s} <_M \mon{t}$. As $<_M$ is a monomial order,
    $\mon{s+r} <_M \mon{t+r}$.
    Hence, $\mon{s} \cdot \mon{r} \prec \mon{t} \cdot \mon{r}$.
  \end{proof}
  
  We extend this property to the homogeneous case.
  
  \begin{corollary} \label{thm:compatMultHom}
    If $\monhom{s}{d_s} \prec \monhom{t}{d_t}$ and
    $\spDeg(\monhom{r}{d_r}) + \spDeg(\monhom{t}{d_t}) =
    \spDeg(\monhom{r}{d_r} \cdot \monhom{t}{d_t})$, then
    $\monhom{s}{d_s} \cdot \monhom{r}{d_r} \prec \monhom{t}{d_t} \cdot
    \monhom{r}{d_r}$.
  \end{corollary}
  
\section{Sparse Gr\"{o}bner Basis (\lowercase{s}GB)}

We want to define and compute \groebner bases in $\K[S_M]$ and
$\K[S_M^h]$ with respect to a (graded) sparse order.
As these orders are not compatible with the multiplication, not all
the standard definitions of \groebner basis are equivalent.
For example, the set of leading monomials of an ideal in $\K[S_M]$
does not necessarily form an ideal.
We say that a set of generators $G$ of an ideal $I \subset \K[S_M]$ is
a sparse \groebner basis with respect to an order $\prec$, if for each
$f \in I$, there is a $g \in G$ such that $\LM_{\prec}(g)$ divides
$\LM_{\prec}(f)$. Similarly for $\K[S_M^h]$.

This definition has a drawback: The multivariate polynomial division
algorithm might not terminate.
This can happen when $\LM_{\prec}(f) = \mon{t} \cdot \LM_{\prec}(g)$ and
$\LM_{\prec}(f) \prec \LM_{\prec}(\mon{t} \cdot g)$.
Then, the reduction step ``increases'' the leading monomial, so that
the algorithm does not necessarily terminates. We can construct
examples where we have a periodic sequence of reductions.
%
%
%
%
To avoid this problem, we redefine the division relation.
   
\begin{definition}[Division relation] \label{def:division}
  For any $\monhom{s}{d_s},\monhom{r}{d_r} \in \K[S_M^h]$, we say that
  $\monhom{s}{d_s}$ divides $\monhom{r}{d_r}$, and write
  $\divs{\monhom{s}{d_s}}{\monhom{r}{d_r}}$, if
  there is a $\monhom{t}{d_t} \in \K[S_M^h]$ such that
  $\monhom{s}{d_s} \cdot \monhom{t}{d_t} = \monhom{r}{d_r}$ and
  $\spDeg(\monhom{s}{d_s}) + \spDeg(\monhom{t}{d_t}) =
  \spDeg(\monhom{r}{d_r})$.
  Similarly, for  $\mon{s},\mon{r} \in \K[S_M]$, we say
  that $\mon{s}$ divides $\mon{r}$, and write
  $\divs{\mon{s}}{\mon{r}}$, if $\divs{\hom(\mon{s})}{\hom(\mon{r})}$.
\end{definition}

\begin{remark} \label{rmk:leadingMonomialAndDivision}
  If $\divs{LM_{\smh}(f)}{\monhom{s}{d_s}}$,
  then there is a $\monhom{t}{d_t} \in \K[S_M^h]$ such that
  $\monhom{s}{d_s} = \monhom{t}{d_t} \cdot LM_{\smh}(f) =  
  LM_{\smh}(\monhom{t}{d_t} \cdot f)$, by \Cref{thm:compatMultAff}. Similarly over $\K[S_M]$.
\end{remark}

We define the sparse \groebner bases (sGB) as follows.

\begin{definition}[sparse \groebner bases] \label{def:sparse-grobner-basis}
  Given a (graded) sparse order $\prec$, see \Cref{eq:ordFromMonomial}, and
  an ideal $I \subset \K[S_M]$, respectively $I \subset \K[S_M^h]$, a
  set $sGB(I) \subset I$ is a sparse \groebner basis (sGB) if it
  generates $I$ and for any $f \in I$ there is some $g \in sGB(I)$ such
  that $\divs{\LM_{\prec}(g)}{\LM_{\prec}(f)}$.
\end{definition}

With this definition, each step in the division algorithm reduces the
leading monomial (\Cref{rmk:leadingMonomialAndDivision}), and so the division
algorithm always terminates, see e.g. \cite[Thm.~2.3.3,Prop.~2.6.1]{cox1992ideals}.

\begin{lemma} \label{thm:divisionConverges}
  Let $f \in \K[S_M]$ and $G$ be a set of polynomials in
  $\K[S_M]$. Using our definition of division relation
  (\Cref{def:division}), the multivariate division algorithm~\cite[Thm.~2.3.3]{cox1992ideals} for the division of $f$ by $G$, with
  respect to the order $\sm$, terminates.
  Moreover, if $G$ is a sGB of an ideal $I$ with respect to $\sm$ and
  $f \equiv f' \mod I$, then the remainder division algorithm for $f$
  and $f'$ is the same and unique for any sGB.
\end{lemma}

\begin{proof}
  By \Cref{rmk:leadingMonomialAndDivision}, each step in the division
  algorithm reduces the leading monomial. The proof follows, mutatis
  mutandis, from~\cite[Thm.~2.3.3,Prop.~2.6.1]{cox1992ideals}.
%
\end{proof}

Our next goal is to prove that for every ideal and sparse order, there
is a finite sGB. A priori, this is not clear from the Noetherian
property of $\K$ as $\LM_{\prec}(I)$ is not an ideal. Our
strategy is to prove that over $\K[S_M^h]$ there is always a finite
sparse \groebner basis, and then extend this result to $\K[S_M]$. We
show that this sGB is related to a standard \groebner basis over some
Noetherian ring, so it is finite.

  \subsection{Finiteness of  sparse \groebner Bases}

  \noindent
  \textbf{Homogeneous case.}

  Let $<_M$ be a monomial order for $\K[S_M]$ and $\sm$ the sparse
  order related to $<_M$, \Cref{eq:ordFromMonomial}. Consider $\smh$
  the graded sparse order related to $\sm$ over $\K[S_M^h]$,
  \Cref{eq:graduateOrder}.
   
  Consider the lattice ideal $T$ from
  \Cref{sec:semigroup-algebra}. This ideal $T$ is homogeneous and the
  algebra $\K[S_M^h]$ is isomorphic to $\K[\bm{y}] / T$ as a graded
  algebra.
  Let $\widetilde{\psi} : \K[\bm{y}] / T \rightarrow \K[S_M^h]$ and
  $\widetilde{\phi} : \K[S_M^h] \rightarrow \K[\bm{y}] / T$ be the
  isomorphisms related to $\K[S_M^h] \cong \K[\bm{y}] / T$, such that
  they are inverse of each other and
  $\widetilde{\psi}(\monhom{0}{1}) = y_0$.
  We extend $\widetilde{\psi}$ to
  $\psi : \K[\bm{y}] \rightarrow \K[S_M^h]$, where
  $\psi(\bm{y}^\alpha)$ is the image, under $\widetilde{\psi}$, of
  $\bm{y}^\alpha$ modulo $T$. The map $\psi$ is a 0-graded
  epimorphism.

  For $\bm{y}^\alpha \in \K[\bm{y}]$, let $\deg(\bm{y}^\alpha,y_0)$ be
  the degree of $\bm{y}^\alpha$ with respect to $y_0$ and
  $\deg(\bm{y}^\alpha)$ be the total degree.
  Given a (standard) monomial order $\widetilde{<}$ for $\K[\bm{y}]$,
  consider the graded monomial order $<_y$ for $\K[\bm{y}]$ defined as
  follows,
  {\small 
    \begin{equation}
      \bm{y}^a <_y \bm{y}^b \iff
      \begin{cases}
        \deg(\bm{y}^a) < \deg(\bm{y}^b) & \!\\
        \deg(\bm{y}^a) = \deg(\bm{y}^b) & \land \quad \deg(\bm{y}^a, y_0) > \deg(\bm{y}^b, y_0) \\
        \deg(\bm{y}^a) = \deg(\bm{y}^b) & \land \quad \deg(\bm{y}^a, y_0) = \deg(\bm{y}^b, y_0) \quad  \land \\
        & \qquad \psi(\bm{y}^{a}) <_M \psi(\bm{y}^{b}) \\
        \deg(\bm{y}^a) = \deg(\bm{y}^b) & \land \quad \deg(\bm{y}^a, y_0) = \deg(\bm{y}^b, y_0) \quad \land \\
        & \qquad \psi(\bm{y}^{a}) = \psi(\bm{y}^{b}) \quad \land \quad \bm{y}^a \, \widetilde{<} \, \bm{y}^b
      \end{cases}
    \end{equation} 
  }
  This order is a monomial order, because it is a total order,
  $\bm{y}^{{0}}$ is the unique smallest monomial (it is
  the only one of degree $0$), and it is compatible with the
  multiplication (every case is compatible).

  For each $f \in \K[\bm{y}]$, we define $\eta$ as the normal form
  (the remainder of the division algorithm) of $f$ with respect to the
  ideal $T$ and the monomial order $<_y$.
  Recall that $\eta = \eta \circ \eta$ and
  $\mathtt{coker}(\eta) \cong \K[\bm{y}] / T$.
  We notice that for each poset in $\K[\bm{y}] / T$, $\eta$ assigns
  the same normal form to all the elements that it
  contains. Therefore, we abuse notation, and we also use $\eta$ to
  denote the map $\K[\bm{y}] / T \to \K[\bm{y}]$ that maps each poset
  to this unique normal form.
  As $T$ is homogeneous, $\eta$ is a 0-graded map.
  We extend $\widetilde{\phi}$ to
  $\phi : \K[S_M^h] \rightarrow \K[\bm{y}]$ as
  $\phi := \eta \circ \widetilde{\phi}$. This map is 0-graded and
  linear, but not a homomorphism.
  It holds $\psi \circ \phi = Id$ and $\phi \circ \psi = \eta$.

  \begin{center}
    \begin{tikzpicture}[->,>=stealth',auto,node distance=6cm,
      thick,main node/.style={font=\sffamily\small\bfseries},every
      node/.style={font=\sffamily\small,align=center,rectangle,
        rounded corners, fill=white}]
      
      \node[main node] (T) {$\K[\bm{y}] / T$};
      \node[main node] (S) [right of=T] {$\K[S_M^h]$};
      \node[main node] (Y) [left of=T]  {$\K[\bm{y}]$};
      
      \path[]
      (Y) edge[bend right=7]  node [below=-0.25] {$\mod T$} (T)
      (T) edge[bend left=-7]  node [above=-0.25] {$\eta$} (Y)
      (S) edge[bend left=-6.5]  node [above=-0.23] {$\widetilde{\phi}$} (T)
      (T) edge[bend right=6.5]  node [below=-0.21] {$\widetilde{\psi}$} (S)
      (Y) edge[bend right=15]  node [below=-0.25] {$\psi$} (S)
      (S) edge[bend left=-15]  node [above=-0.25] {$\phi$} (Y);
    \end{tikzpicture} 
  \end{center}
  
  \begin{theorem} \label{thm:homogeneousStandardSparse}
    Let $I^h \subset \K[S_M^h]$ be a homogeneous ideal and consider the
    homogeneous ideal
    ${J}^h := \langle \phi(I^h) + T \rangle \subset \K[\bm{y}]$. If the
    \groebner base of $J^h$ with respect to $<_y$ is $GB_{<_y}(J^h)$,
    then $\psi(GB_{<_y}(J^h))$ is a sparse \groebner base of $I^h$ with
    respect to $\smh$.
  \end{theorem}

  To prove the theorem we need the following lemmas.

  \begin{lemma} \label{thm:spDegAndY0}
    For all $\bm{y}^\alpha \in \K[\bm{y}]$,
    $\small \deg(\eta(\bm{y}^\alpha), y_0) = \deg(\bm{y}^{\alpha}) -
    \spDeg(\psi(\bm{y}^\alpha))$.
  \end{lemma}

  \begin{proof}
    Let $\monhom{s}{d} := \psi(\bm{y}^\alpha)$ and
    $\bar{d} = \spDeg(\monhom{s}{d})$.
    Note that $d = \deg(\bm{y}^\alpha)$, because $\psi$ is 0-graded.
    We can write
    $\psi(\bm{y}^\alpha) = \hom(\mon{s}) \cdot
    \monhom{0}{d - \bar{d}}$.
    Recall that $\phi \circ \psi = \eta$.
    Applying $\phi$ to the previous equality we get,
    $\eta(\bm{y}^\alpha) = \eta(\bar{\phi}(\hom(\mon{s})) \cdot
    \bar{\phi}(\monhom{0}{d-\bar{d}})) =
    \eta(\bar{\phi}(\hom(\mon{s})) \cdot y_0^{d - \bar{d}})$.
    Note that the order $>_y$ acts as the degree reverse
    lexicographical with respect to $y_0$, hence
    $\eta(\bar{\phi}(\hom(\mon{s})) \cdot y_0^{d - \bar{d}}) =
    \phi(\hom(\mon{s})) \cdot y_0^{d - \bar{d}}$.
    If $y_0$ divides $\phi(\hom(\mon{s}))$, then there is a monomial
    $\bm{y}^\beta$ such that
    $y_0 \cdot \bm{y}^\beta = \phi(\hom(\mon{s}))$, and so,
    $\psi(y_0 \cdot \bm{y}^\beta) =
    \psi(\phi(\hom(\mon{s})))$.
    As $\psi \circ \phi = Id$ and $\psi$ is a $0$-graded epimorphism,
    then $\monhom{0}{1} \cdot \psi(\bm{y}^\beta) = \hom(\mon{s})$, but
    this is not possible by definition of homogenization
    (\Cref{sec:sparse-degree-and-homogenization}).
    Hence, $\deg(\phi(\hom(\mon{s})), y_0) = 0$ and
    $\deg(\eta(\bm{y}^\alpha),y_0) = 0 + d - \bar{d}$.
  \end{proof}
  
  \begin{corollary} \label{thm:spDegAndY0versionPolytope}
    For all $\monhom{s}{d} \in \K[S_M^h]$, it holds $$\spDeg(\monhom{s}{d}) = d -
    \deg(\phi(\monhom{s}{d}),y_0).$$
  \end{corollary}

  As $\psi$ and $\phi$ are $0$-graded maps, by \Cref{thm:spDegAndY0}
  and \Cref{thm:spDegAndY0versionPolytope}, they preserve the order.

  \begin{corollary} \label{thm:etaSmallerImpliesPsiSmaller}
    $\eta(y^{\alpha}) <_y \eta(y^{\beta}) \implies \psi(y^{\alpha}) \smh
    \psi(y^\beta)$.
  \end{corollary}


  
  \begin{lemma}  \label{thm:divisionOnYandOnPolytope}
    $\bm{y}^{\alpha} | \phi(\monhom{s}{d}) \implies \divs{\psi(\bm{y}^{\alpha})
    }{\monhom{s}{d}}$.
  \end{lemma}

   \begin{proof}
     Let $\bm{y}^\beta$ such that $\bm{y}^{\alpha} \cdot
     \bm{y}^{\beta} = \phi(\monhom{s}{d})$, so $\psi(\bm{y}^{\alpha})
     \cdot \psi(\bm{y}^{\beta}) = \monhom{s}{d}$.
     As $\eta$ is a normal form,
     $\eta(\phi(\monhom{s}{d})) = \phi(\monhom{s}{d})$ and then,
     $\eta(\bm{y}^\alpha) = \bm{y}^\alpha$ and
     $\eta(\bm{y}^\beta) = \bm{y}^\beta$.
     Hence, by \Cref{thm:spDegAndY0versionPolytope},
     $\spDeg(\psi(\bm{y}^{\alpha} \cdot \bm{y}^{\beta})) =
     \deg(\bm{y}^{\alpha} \cdot \bm{y}^{\beta}) -
     \deg(\eta(\bm{y}^{\alpha} \cdot \bm{y}^{\beta}),y_0) =
     \deg(\bm{y}^{\alpha}) - \deg(\eta(\bm{y}^{\alpha}), y_0) +
     \deg(\bm{y}^{\beta}) - \deg(\eta(\bm{y}^{\beta}), y_0) =
     \spDeg(\psi(\bm{y}^\alpha)) + \spDeg(\psi(\bm{y}^\beta))$, by \Cref{thm:spDegAndY0}.
   \end{proof}
       
   \begin{corollary}
     \label{thm:divisionLeadingImpliesDivisionLeading}
     For all $f \in \K[S_M^h]$, for all $g \in \K[\bm{y}]$, it holds
     $$\LM_{<_y}(\eta(g)) | \LM_{<_y}(\phi(f)) \implies
     \divs{\LM_{\smh}(\psi(g))}{\LM_{\smh}(f)}.$$
   \end{corollary}

   \begin{proof}
     By \Cref{thm:etaSmallerImpliesPsiSmaller},
     {\small$\psi(\LM_{<_y}(\eta(g))) = \LM_{\smh}(\psi(g))$} and 
     {\small$\psi(\LM_{<_y}(\phi(f)))$} \linebreak $= \LM_{<_y}(\psi(\phi(f))) =
     \LM_{\smh}(f)$. The proof follows from
     \Cref{thm:divisionOnYandOnPolytope}.
   \end{proof}
   
   \begin{proof}[Proof of \Cref{thm:homogeneousStandardSparse}]
     Consider $f \in I^h$, then $\phi(f) \in J^h$. Hence, there are
     $g_1,\dots,g_k \in GB_{<_y}(J^h)$ and
     $p_1,\dots,p_k \in \K[\bm{y}]$ such that
     $\phi(f) = \sum_{i=1}^k p_i \cdot g_i$.
     As $\psi \circ \phi = Id$ and $\psi$ is an epimorphism such that
     $\psi(T) = 0$, then
     $\psi(\phi(f)) = f = \sum_{i=1}^k \psi(p_i) \cdot \psi(g_i)$ and
     $\psi(g_i),\dots,\psi(g_k) \in I^h$.
     Hence, $\psi(GB_{<_y}(J^h))$ generates $I^h$.
     
     The set $GB_{<_y}(J^h)$ is a \groebner basis, then there is a
     $g \in GB_{<_y}(J^h)$ such that
     $\LM_{<_y}(g) | \LM_{<_y}(\phi(f))$.
     As $\phi(f) = \eta(\phi(f))$,
     $\eta(\LM_{<_y}(\phi(f))) = \LM_{<_y}(\phi(f))$ and
     $\eta(\LM_{<_y}(g)) = \LM_{<_y}(g)$.
     As $\eta$ is a normal form wrt $<_y$,
     $\eta(\LM_{<_y}(g)) = \LM_{<_y}(\eta(g))$.
     By
     \Cref{thm:divisionLeadingImpliesDivisionLeading},
     $\divs{\LM_{\smh}(\psi(g))}{\LM_{\smh}(f)}$. Hence,
     $\psi(GB_{<_y}(J^h))$ is a sGB for $I^h$ with
     respect to $\smh$.
   \end{proof}
   
   \begin{corollary}
     \label{thm:finiteHomogeneousSGB}
     Given an ideal $I^h \subset \K[S_M^h]$ and a graded sparse order
     $\smh$, its sGB with respect to this order is
     finite.
   \end{corollary}
     
   \begin{proof}
     In \Cref{thm:homogeneousStandardSparse} we construct
     $sGB_{\smh}(I^h)$ from a (standard) \groebner basis of an
     ideal of $\K[\bm{y}]$, finite as $\K[\bm{y}]$ is Noetherian.
   \end{proof}


   \noindent
   \textbf{Non-homogeneous case.}
   Let $\sm$ be a sparse order for $\K[S_M]$.

   \begin{lemma}
     \label{thm:dehomGBofHomIdealIsGB}
     Let $I^h \subset \K[S_M^h]$ be a homogeneous ideal. Let $\smh$
     be the graded sparse order for $\K[S_M^h]$ related to
     $\sm$. Then, $\dehom(sGB_{\smh}(I^h))$ is a sparse \groebner
     Basis for $\dehom(I^h)$ with respect to $\sm$.
   \end{lemma}

   \begin{proof}
     The set $\dehom(sGB_{\smh}(I^h))$ generates $\dehom(I^h)$.
     Note that for homogeneous polynomials, $\LM_{\smh}$ commutes with
     the dehomogenization, that is for any homogeneous polynomial
     $g \in \K[S_M^h]$, $\LM_{\sm}(\dehom(g)) = \dehom(\LM_{\smh}(g))$.
%
%
     Consider $\bar{f} \in \dehom(I^h)$, then there is an $f \in I^h$
     such that $f = \dehom(\bar{f})$.
     In addition, there is $g \in sGB_{\smh}(I^h)$ such that
     $\divs{\LM_{\smh}(g)}{\LM_{\smh}(f)}$. Let
     $\monhom{s}{d} \in \K[S_M^h]$ such that
     $\LM_{\smh}(g) \cdot \monhom{s}{d} = \LM_{\smh}(f)$ and
     $\spDeg(\LM_{\smh}(g)) + \spDeg(\monhom{s}{d}) =
     \spDeg(\LM_{\smh}(f))$.
     The sparse degree $\spDeg$ is independent of the homogeneous
     degree, so
     $\spDeg(\dehom(\LM_{\smh}(g))) + \spDeg(\mon{s}) =
     \spDeg(\dehom(\LM_{\smh}(f)))$. Hence,
     $\spDeg(\LM_{\sm}(\dehom(g))) + \spDeg(\mon{s}) =
     \spDeg(\LM_{\sm}(\bar{f}))$ and
     $\LM_{\sm}(\dehom(g)) \cdot \mon{s} = \LM_{\sm}(\bar{f})$, so \linebreak
     $\divs{\LM_{\sm}(\dehom(g))}{\LM_{\sm}(\bar{f})}$ and
     $\dehom(sGB_{\smh}(I^h))$ is a sGB of $\dehom(I^h)$ wrt $\sm$.
   \end{proof}

   \begin{corollary} \label{thm:finiteAffineSGB}
     The  sGB of $I \subset \K[S_M]$ with respect
     to $\sm$ is finite.
   \end{corollary}
   \begin{proof}
     For $\hom(I)$, the homogenization of $I$, $\dehom(\hom(I)) = I$.
     So by \Cref{thm:dehomGBofHomIdealIsGB}
     $\dehom(sGB_{\prec}(\hom(I)))$ is a sGB of $I$ and is finite by
     \Cref{thm:finiteHomogeneousSGB}.
   \end{proof}

   \subsection{Computing sparse \groebner Bases}

   \noindent
   \textbf{Homogeneous case.}
   To compute a sGB of a homogeneous ideal
   $I^h := \langle f_1,\dots,f_k \rangle$ with respect to $\smh$, we
   introduce the $D$-sparse \groebner bases
   \cite[Sec.~III.B]{lazard1983grobner}. A $D$-sparse \groebner
   basis of $I^h$ is a finite set of polynomials $\mathcal{J}^h \subset I^h$
   such that for each $f \in I^h$ with $\degHom(f) \leq D$, it holds
   $f \in \langle \mathcal{J}^h \rangle$ and there is a $g \in \mathcal{J}^h$ such that
   $\divs{\LM_\smh(g)}{\LM_{\smh}(f)}$. For big enough $D$, for
   example equal to the maximal degree in the polynomials in
   $sGB_\smh(I^h)$, a $D$-sparse \groebner basis is a sparse
   \groebner basis. The \emph{witness degree} of $I^h$ is the minimal
   $D$ such that a $D$-sparse \groebner basis is a sGB.
   We compute $D$-sparse \groebner bases by using linear algebra.
   
   \begin{definition} \label{thm:defMacaulayMatrix}
     A Macaulay matrix $\mathcal{M}$ is a matrix whose columns are
     indexed by monomials in $\K[S_M^h]$ and the rows by polynomials
     in $\K[S_M^h]$.
     The set of monomials that index the columns contain all the
     monomial in the supports of the polynomials of the rows.
     For a monomial $m$ in a polynomial $f$, the entry in the matrix
     indexed by $(m,f)$ is the coefficient of the monomial $m$ in $f$.
     We define $\mathtt{Columns}(\mathcal{M})$ as the sequence of the
     monomials of $\mathcal{M}$ in the order that they index the
     columns.
     We define $\rows(\mathcal{M})$ as the set of \emph{non-zero}
     polynomials that index the rows of $\mathcal{M}$.
   \end{definition}
     
   If we apply a row operation to a Macaulay matrix, we obtain a new
   Macaulay matrix, where we replace one of the polynomials (that is
   one of the rows) by linear combinations of some of them.
   We say that we have a \emph{reduction to zero}, if after we perform
   a row operation, the resulting row is zero.
   As observed by~\citet{lazard1983grobner}, if we sort the columns in
   decreasing order by $\smh$, we can compute a \groebner basis
   using Gaussian elimination.  The proof of the following lemma
   follows from \cite{lazard1983grobner}.

   \begin{lemma} \label{thm:computeWithMacaulayMatrix}
     Consider the ideal
     $I^h := \langle f_1,\dots,f_k \rangle \subset \K[S_M^h]$.
     Let $\M_D$ be the Macaulay matrix whose columns are all the
     monomials in $\K[S_M^h]_D$ sorted in decreasing order by $\smh$,
     and the rows are all the products of the form
     $\monhom{s}{D - \degHom(f_i)} \cdot f_i \in
   \K[S_M^h]_{D}$.
   Let $\widetilde{\M_D}$ be the matrix obtained by applying Gaussian
   elimination to $\M_D$ to obtain a reduced row echelon form.
     Then, the polynomials in $\bigcup_{i=1}^D \rows(\widetilde{\M_i})$
     form a $D$-sparse \groebner basis.
     Moreover, if we only consider the set of polynomials whose
     leading monomial can not be divided by the leading monomial of a
     polynomial obtained in smaller degree, that is
     \[ \bigcup_{i=1}^D \{f \in \rows(\widetilde{\M_i}) :
       (\nexists \, g \in \bigcup_{j=1}^{i-1} \rows(\widetilde{\M_j})) \;
       \divs{\LM_\smh(g)}{\LM_\smh(f)} \}, \]
     then this subset is a $D$-sparse \groebner basis too.
   \end{lemma}

   \noindent
   \textbf{Non-homogeneous case.}

   Given an ideal
   $I := \langle \bar{f_1}\dots \bar{f_r} \rangle \subset \K[S_M]$,
   we homogenize the polynomials and use
   \Cref{thm:computeWithMacaulayMatrix} to compute a sparse
   \groebner basis with respect to $\smh$. By
   \Cref{thm:dehomGBofHomIdealIsGB}, if we dehomogenize the computed basis, we
   obtain a sparse \groebner basis with respect to $\sm$ of $I$.
   Instead of homogenizing all polynomials $\bar{f_i}$ simultaneously,
   we consider an iterative approach, which, under regularity
   assumptions, involves only full-rank matrices, and hence avoids all
   reductions to zero.
   The following lemma allows us to compute a sparse \grobner basis in the
   homogeneous case, from the non-homogeneous one.
   \begin{lemma}
     If $G$ is a sGB of $I$ with respect to $\sm$, then
     $G^{h} := \hom(G)$ is a sGB of
     $\langle \hom(I) \rangle$ with respect to $\smh$.
   \end{lemma}

   \begin{proof}
     First note that the homogenization commutes with the leading
     monomial, that is $\forall \bar{g} \in \K[S_M]$,
     $\LM_{\smh}(\hom(\bar{g})) \!=\! \hom(\LM_{\sm}(\bar{g}))$.
     Let $f \in \langle \hom(I) \rangle$. We can write $f$ as
     $\monhom{0}{\degHom(f) - \spDeg(f)} \cdot
     \hom(\dehom(f))$. Consider $\bar{g} \in G$ such that
     $\divs{\LM_\sm(\bar{g})}{\LM_\sm(\dehom(f))}$.
     By definition (\Cref{def:division}),
     $\divs{\hom(\LM_\sm(\bar{g}))}{\hom(\LM_\sm(\dehom(f)))}$, and by
     commutativity, it holds that
     $\divs{\LM_\smh(\hom(\bar{g}))}{\LM_\smh(\hom(\dehom(f))}$.
     The sparse degree and the leading monomials with respect to
     $\smh$ are invariants under the multiplication by
     $\monhom{0}{1}$. Hence,
     $\divs{\LM_\smh(\hom(\bar{g}))}{\LM_\smh(f)}$.
     To conclude, we have to prove that $G^{h}$ is a basis of
     $\langle \hom(I) \rangle$. As for each $f \in \hom(I)$ there is a
     $\bar{g} \in G$ such that
     $\divs{\LM_\smh(\hom(\bar{g}))}{\LM_\smh(f)}$.
     Thus, the remainder of the division algorithm
     (\Cref{thm:divisionConverges}) is zero, and so we obtain a
     representation of $f$ in the basis $\hom(G)$.
   \end{proof}

   \begin{corollary}
     \label{thm:algorithmCorrect}
     Let $I \subset \K[S_M]$ be an (non-homogeneous) ideal and consider the
     (non-homogeneous) polynomial $\bar{f} \in \K[S_M]$. Let $G$ be a (non-homogeneous)
     sGB of $I$ wrt $\sm$ and $G^h_{\bar{f}}$ be a (homogeneous) sGB
     of $\langle \hom(G) + \hom(\bar{f}) \rangle$ wrt $\smh$. Then,
     $\dehom(G^h_{\bar{f}})$ is a (non-homogeneous) sGB of $\langle I + \bar{f} \rangle$
     wrt $\sm$.
   \end{corollary}

   \Cref{thm:algorithmCorrect} supports an iterative algorithm to
   compute a sGB of $I$.  For each $i \leq n$, let
   $I_i := \langle \bar{f}_1, \dots, \bar{f}_i \rangle$ and
   $G_i := sGB_{\sm}(I_i)$. Consider
   $I_i^h := \langle \hom(G_{i-1}) + \hom(\bar{f_i}) \rangle$. By
   \Cref{thm:algorithmCorrect}, we can consider $G_i$ as
   $\dehom(sGB_{\smh}(I_i^h))$. To compute $sGB_{\smh}(I_i^h)$ we use
   \Cref{thm:defMacaulayMatrix}.

   Many rows of the Macaulay matrices reduces to zero during the
   Gaussian elimination procedure. We can adapt the F5 criterion
   \cite{faugere2002F5,eder2014survey} to identify these rows and
   avoid them.

   \begin{lemma}
     \label{thm:givenGBformMacMatrix}
     Let $G$ be a sGB of the homogeneous ideal
     $I^h$ wrt $\smh$.
     Let $\mathcal{N} \subset \K[S^h_M]_D$ be the set of monomials of
     degree $D$ such that for each of them there is a polynomial in
     $G$ whose leading term divides it, that is
     $ \mathcal{N}= \left\{ \monhom{s}{D} \in \K[S_M^h]_D \,:\,
       \exists \, g \in G \text{ s.t. }
       \divs{\LM_\smh(g)}{\monhom{s}{D}} \right\} .  $
     To each $\monhom{s}{D} \in \mathcal{N}$ associate only one
     polynomial $g \in G$, such that
     $\divs{\LM_\smh(g)}{\monhom{s}{D}}$. Let $\mathcal{R}$ be the set
     formed by the polynomials
     $\frac{\monhom{s}{D}}{\LM_\smh(g)} \cdot g$ where $g$ is the
     polynomial associated to $\monhom{s}{D} \in \mathcal{N}$.

     Consider the Macaulay matrix $\M'_D$ with columns indexed by the monomials in
     $\K[S_M^h]_D$ in decreasing order  w.r.t. $\smh$ and rows
     indexed by $\mathcal{R}$.
     Let $\widetilde{\M'_d}$ be the Macaulay matrix obtained after
     applying Gaussian elimination to $\M'_d$ to obtain a reduced row
     echelon form.
     Then, $\rows(\widetilde{M'_D}) = \rows(\widetilde{M_D})$, where
     $\widetilde{M_d}$ is the Macaulay matrix of
     \Cref{thm:computeWithMacaulayMatrix} with respect to
     $G^h$. Moreover, the matrix $M'_D$ is full-rank and in row echelon
     form.
   \end{lemma}

   \begin{proof}
     By construction, we are skipping the polynomials whose leading
     monomials already appear in $M'_D$. Hence, each row has a
     different leading monomial and so, the matrix $M'_D$ is
     full-rank.
     If we add to $M'_D$ a new homogeneous polynomial of degree $D$
     belonging to the ideal $I^h$, then it must be linear dependent with
     the polynomials in $\rows(M'_D)$. If not, after reducing the
     polynomial by the previous rows, we discovered a new polynomial
     in the ideal $I^h$ with a leading monomial which is not divisible
     by $G^h$. But this is not possible because $G^h$ is a sparse
     \groebner basis.
   \end{proof}

   \begin{lemma}[Sparse F5 criterion] \label{thm:sparse-f5}
     Let $G^h$ be a sparse \groebner basis of the
     homogeneous ideal $I^h$ wrt $\smh$ and let $\M'_D$ be
     the Macaulay matrix of \Cref{thm:givenGBformMacMatrix} of degree
     $D$.
     Let $d \in \N$ and consider the set 
     $
       \b = \{
       \monhom{s}{D - d} \in \K[S_M^h]_{D-d}  \,:\, 
       \nexists \, g \in G^h \text{ s.t. }  \divs{\LM_\smh(g)}{\monhom{s}{D - d}}
       \} .
     $
     \noindent
     Let $f \in \K[S_M^h]_d$; consider the Macaulay matrix
     $\M^*_D$ obtained after appending to $\M'_D$ rows indexed by
     $\{\monhom{s}{D - d} \cdot f  \,:\, \monhom{s}{D - d} \in \b\}.$

     Let $\widetilde{\M^*_D}$ be the matrix obtained after applying
     Gaussian elimination to $\M^*_D$.
     Then, $\rows(\widetilde{\M^*_D}) = \rows(\widetilde{\M_D})$,
     where $\widetilde{\M_D}$ is the Macaulay matrix of
     \Cref{thm:computeWithMacaulayMatrix} for the ideal
     $\langle G^h, f \rangle$.
     Moreover, if $f$ is not a zero-divisor in $\K[S_M^h] / I^h$, then
     $\M^*_D$ is full-rank.
   \end{lemma}

   \begin{proof}
     Let $\monhom{s}{D - d}$ be a monomial such that there is
     a $g \in G^h$ such that
     $\divs{\LM_\smh(g)}{ \monhom{s}{D - d} }$.
     Consider
     $p := \frac{\monhom{s}{D - d}}{\LM_\smh(g)} \cdot g$. By
     \Cref{rmk:leadingMonomialAndDivision},
     $\LM_\smh(p) = \monhom{s}{D - d}$. Consider
     $p_{red} := LT_\smh(h) + q$, where $q$ is the remainder of the
     division of $p - LT_\smh(p)$ by $G^h$. It holds $p_{red} \in I^h$.
     Also all the monomials in the support of  $q$ are not divisible by
     the leading monomials of $G^h$
     (\Cref{thm:divisionConverges}). Then, using the rows of $\M^*_D$
     we can form the polynomial $f \cdot q$. If we add the row
     corresponding to $f \cdot \monhom{s}{D - d}$, we can
     reduce this polynomial to zero as
     $f \cdot \monhom{s}{D - d} + f \cdot q = f \cdot p_{red}
     \in I^h$.
     If $f$ is not a zero-divisor in $\K[S^h_M] / I^h$, then
     $g \cdot f \in I^h$, implies $g \in I^h$ and so
     $LM_{\smh}(g) \in G^h$. Hence, we skip every row reducing to zero
     involving $f$.
   \end{proof}
   
   \begin{lemma}
     If $\bar{f}_1,\dots,\bar{f}_k \in \K[S_M]$ is a regular sequence, then for each
     $i \leq k $, $\hom(\bar{f}_i)$ is not a zero-divisor of
     $\K[S_M^h] / \hom(\langle \bar{f}_1,\dots,\bar{f}_{i-1} \rangle)$.
   \end{lemma}

   \begin{proof}
     If $\hom(\bar{f}_i)$ is a zero-divisor of
     $\K[S_M^h] / \hom(\langle \bar{f}_1,\dots,\bar{f}_{i-1} \rangle)$, there is a
     $g \in \K[S_M^h]$ such that
     $g \not\in \hom(\langle \bar{f}_1,\dots,\bar{f}_{i-1} \rangle)$ and
     $g \cdot \hom(\bar{f}_i) \in \hom(\langle \bar{f}_1,\dots,\bar{f}_{i-1}
     \rangle)$. By definition of the homogenization of an ideal,
     $\dehom(g) \not\in \langle \bar{f}_1,\dots,\bar{f}_{i-1} \rangle$ but, as
     $\dehom$ is a homomorphism,
     $\dehom(g) \cdot \bar{f}_i \in \langle \bar{f}_1,\dots,\bar{f}_{i-1}
     \rangle$. So, $\bar{f}_1,\dots,\bar{f}_i$ is not a regular sequence.
   \end{proof}

   Hence, given the witness degrees of each $I_i^h$, we have the
   algorithm \Cref{alg:matrixF5} to compute iteratively a sparse \groebner
   basis.

   As in the standard case, we can define the reduced sGB and adapt
   \cite[Prop.~2.7.6]{cox1992ideals} to prove their finiteness and
   uniqueness.

      \begin{algorithm}
     \caption{$\texttt{M}^2$: Mixed sparse Matrix-F5 with respect to $\sm$}
     \begin{algorithmic}
     \label{alg:matrixF5}
     \REQUIRE \parbox[t]{180px}{$\bar{f}_1,\dots,\bar{f}_k \in \K[S_M]$
       and $d^{wit}_1,\dots,d^{wit}_k$ such that $d^{wit}_i$ is the
       witness degree of $I^h_i$.}
%
         \FOR{$i = 1$ to $k$}
           \STATE $G_i \leftarrow \emptyset$
           \FOR{$d = 1$ to $d^{wit}_i$}
           \STATE $\M^i_d \leftarrow$ \parbox[t]{300px}{Macaulay matrix
             with columns indexed by the monomials in $\K[S_M^h]_d$
             in decreasing order by $\smh$}
             \FOR{$\monhom{s}{d} \in \K[S_M^h]_d$}
               \IF{$\exists g \in G^h_{i-1} : \divs{\LM_{\smh}(g)}{\monhom{s}{d}}$}
                 \STATE Add to $\M^i_d$ the polynomial $\frac{\monhom{s}{d}}{\LM_{\smh}(g)} \cdot g$
               \ENDIF
             \ENDFOR
             \FOR{$\monhom{s}{d-\degAff(\bar{f}_i)} \in \K[S_M^h]_{d-\degAff(\bar{f}_i)}$}
               \IF{$\nexists g \in G^h_{i-1}$ such that $\divs{\LM_{\smh}(g)}{\LM_{\smh}(\hom(\bar{f}_i))}$}
                 \STATE Add to $\M^i_d$ the polynomial $\monhom{s}{d-\degAff(\bar{f}_i)} \cdot \hom(\bar{f}_i)$
               \ENDIF
             \ENDFOR
             \STATE $\widetilde{\M^i_d} \leftarrow$ Gaussian elimination of $\M^i_d$
             \STATE $G_i \leftarrow G_i \cup \{ \bar{h} \in \dehom(\rows(\widetilde{\M^i_d})) : 
             \nexists\, \bar{g} \in G_i \land \divs{\LM_{\sm}(\bar{g})}{\LM_{\sm}(\bar{h})} \}$

           \ENDFOR
           \STATE $G^h_{i} \leftarrow \hom(G_{i})$ %
           \ENDFOR

           \RETURN $G_k$
     \end{algorithmic}
   \end{algorithm}

   \section{Multihomogeneous systems}

   We consider an algorithm for solving 0-dimensional square
   multihomogeneous systems with no solutions at infinity.

   \noindent
   \textbf{Notation.}
   \label{sec:Multihom-notation}
   Let $n_1,\dots n_r \in \N$, $N := \sum_i n_i$, and
   $\bm{n} := (n_1 \! \dots \! n_r)\!\in\!\N^r$.
   For $1 \leq i \leq r$, let
   $\bm{x_i}$ be the set of variables $\{x_{i,0},\dots,x_{i,n_i}\}$.
   Let $\K[\bm{x}] := \bigotimes_{i=1}^r \K[\bm{x}_i]$ be the
   multihomogeneous $\K$-algebra multigraded by $\Z^r$, such that for
   all $\bm{d} := (d_1,\dots,d_r) \in \Z^r$, we have
   $\K[\bm{x}]_{\bm{d}} := \bigotimes_{i=1}^r \K[\bm{x}_i]_{d_i}$.
   Given a $\K[\bm{x}]$-module $\Modu$, we consider $[\Modu]_{\bm{d}}$ as the
   graded part of $\Modu$ of multidegree $\bm{d}$.
   Given two multidegrees $\bm{d}$ and $\bar{\bm{d}}$, we say that
   $\bm{d} \geq \bar{\bm{d}}$ if the inequality holds component-wise.
   We consider the multiprojective space
   $\P := \Pr^{n_1} \times \cdots \times \Pr^{n_r}$.

   Let $\bm{\bar{1}} = (1,\dots,1) \in \Z^r$ be the multidegree 
   corresponding to multilinear polynomials in $\K[\bm{x}]$.
   Let $B = \cap_{i = 1}^r \langle x_{i,0},\dots,x_{i,n_i} \rangle$ be
   the ideal generated by all the polynomials in
   $\K[\bm{x}]_{\bm{\bar{1}}}$.

   Consider multihomogeneous polynomials
   $f_1,\dots,f_k \in \K[\bm{x}]$ and denote their multidegrees by
   $\deg(f_1),\dots,\deg(f_k) \in \Nz^r$.
   Let $V_\P(f_1,\dots,f_k)$ be the zero set of $f_1,\dots,f_k$ over
   $\P$. If the dimension of $V_\P(f_1,\dots,f_k)$ over $\P$ is
   $N - k$, then the polynomials $f_1,\dots,f_k$ form a regular
   sequence at each point of $\Pr^1 \times \cdots \times \Pr^r$.
   That is, for each prime ideal $\mathfrak{p}$,  such that $\mathfrak{p} \not\subset B$,
   $(f_1,\dots,f_k)$ form a regular sequence over
   $\K[\bm{x}]_{\mathfrak{p}}$, the localization of $\K[\bm{x}]$ at
   $\mathfrak{p}$.
   In this case, we say that $(f_1,\dots,f_k)$ is a \emph{regular sequence
   outside $B$}.
   This kind of sequence is related to the filter regular
   sequence~\cite[Sec.~2]{trung1998castelnuovo} and the sequence of
   ``almost'' nonzero divisors~\cite[Sec.~3]{maclagan_multigraded_2004},
   \cite[Sec.~2]{sidman2006multigraded}.
 
   Let $\mathcal{K}_\bullet(f_1,\dots,f_k \, ; \, \K[\bm{x}])$ be the
   Koszul complex of $f_1,\dots,f_k$ over $\K[\bm{x}]$. Let
   $H_i(\mathcal{K}_\bullet(f_1,\dots,f_k \, ; \, \K[\bm{x}]))$ be the
   $i$-th Koszul homology module. We also write this homology module as
   $H_i^k$. 

   Let $\bm{x}_h := \prod_{i=1}^r x_{i,0} \in \K[\bm{x}]_{\bm{\bar{1}}}$.
%
   We say that a multihomogeneous system $(f_1,\dots,f_N)$ has
   \emph{no solutions} at infinity if the system
   $(f_1,\dots,f_N, \bm{x}_h)$ has no solutions over $\P$.
   We dehomogenize a multihomogeneous polynomial by replacing each
   variable $x_{i,0}$ with $1$. Let $\K[\bar{\bm{x}}]$ be the
   $\K$-algebra obtained by the dehomogenization of
   $\K[\bm{x}]$. Given $f \in \K[\bm{x}]$, we consider
   $\bar{f} \in \K[\bar{\bm{x}}]$, its dehomogenization.

   \begin{remark}
     There is a (multigraded) isomorphism between the
     multihomogeneous $\K$-algebra $\K[\bm{x}]$ and the polytopal
     algebra $\K[S_{M_1,\dots,M_r}^h]$, where
     $M_i$ are cross products of simplex polytopes.
   \end{remark}
 
   \subsection{Multigraded regularity}
   \label{sec:Cast-Mum-regularity}

   Based on~\citet{maclagan2003uniform,maclagan_multigraded_2004},
   \citet{botbol_castelnuovo_2017} define the multigraded
   Castelnuovo-Mumford regularity over $\K[\bm{x}]$ in terms of the
   vanishing of the \emph{local cohomology} modules with respect to
   $B$.  For an introduction to local cohomology, we refer to~\cite{brodmann_local_2013}.  In the following we present some
   results from~\cite[Chp.~6]{botbol_implicitization_2011}, that we
   need in our setting, see also~\cite{awane2005formes}.
   
   Given a module $\Modu$, $H_B^j(\Modu)$ is the $j$-th local
   cohomology module at $B$ and 
   $\mathtt{sp}(\Modu) := \{\bm{d} \in \Z^r : [\Modu]_{\bm{d}} \neq 0\}$
   is the set of
   multidegrees where the module is not zero.
   In
   \cite{botbol_castelnuovo_2017,botbol_implicitization_2011},
   $\mathtt{sp}(M)$ is called the support of $M$.
   
   Consider $\alpha \subset \{1,\dots,r\}$. We define the $Q_\alpha$
   as the convex region of $\R^r$ given by the vectors
   $(v_1,\dots,v_r) \in \R^r$ so that for every $i \leq r$, 
   \[
   \begin{cases}
     v_i \leq - n_i - 1 & \text{, if } i \in \alpha \\
     v_i \geq 0 & \text{, otherwise} .
   \end{cases}   
   \]
   
   Consider the multiset 
   $\bm{\Sigma}_i^k := \{ \sum_{j \in I} \deg(f_j) : I \subset
   \{1 \dots k\}, \#I \!=\!i\}$ containing the sums of the degrees of $i$
   (different) polynomials from the set $\{f_1,\dots,f_k \}$. 
   Given $v \in \R^r$, the displacement of $Q_\alpha$ by $v$ is 
   $Q_\alpha + v := \{w \in \R^r : w - v \in Q_\alpha\}$.
   Let $N_\alpha := \sum_{i \in \alpha} n_i$.

   \begin{lemma}[{\cite[Lem.~6.4.7]{botbol_implicitization_2011}, ,
       \cite[Prop.~4.2]{awane2005formes}}] \label{thm:regularityR}
     If
      $\mu \not\in \bigcup_{\substack{\alpha \in \{1,\dots,r\} \\ N_\alpha + 1 = l}}
      Q_\alpha$, then $(H_B^l(\K[\bm{x}]))_\mu = 0$.  Equivalently, 
     $$\mathtt{sp}(H_B^l(\K[\bm{x}])) \subset \bigcup_{\substack{\alpha \in \{1,\dots,r\} \\ N_\alpha + 1 = l \\ \alpha \neq \emptyset}} Q_\alpha.$$
   \end{lemma}
   
   \begin{proposition}[{\cite[Remark
       6.4.10]{botbol_implicitization_2011},
       \cite[Cor.~4.3]{awane2005formes}}] \label{thm:bound_reg}
     If $(f_1,\dots,f_k)$ form a regular sequence outside $B$, for
     every $i,j$,
     \begin{align} \label{eq:localCohom}
     \mathtt{sp}(H^i_B(H_j^k)) \subset
      \bigcup_{\substack{\alpha \subset \{1,\dots,k\} \\ N_\alpha + 1 + j - i \leq k \\ \alpha \neq \emptyset}}
       \bigcup_{v \in \bm{\Sigma}^k_{N_\alpha+1+j-i}} Q_\alpha + v.
     \end{align}
   \end{proposition}

  \begin{proof}
    As we assume that $f_1,\dots,f_k$ form a regular sequence outside
    $B$, we have that $H^w_B(H_j^k) = 0$ for all $w > 0$. Hence, the
    \emph{cohomological dimension} of $H_j^k$ with respect to $B$ is
    $0$. Therefore, by \cite[Rmk.~6.2.5 and
    Thm.~6.2.4]{botbol_implicitization_2011},
    $\Supp(H^0_B(H_j^k)) \subset \bigcup\limits_{i \in \Z}
    \Supp(H^i_B(\mathcal{K}^k_{i+j}))$. By definition
    $\mathcal{K}^k_{i+j} = 0$, for $i+j > k$ and
    $\mathcal{K}^k_{i+j} = \bigoplus_{v \in \bm{\Sigma}^k_{i+j}}
    \K[\bm{x}](-v)$, where $\K[\bm{x}](-v)$ is the twist (shift)
    of $\K[\bm{x}]$ by $-v$.
    Hence,
    {\small
      $$\mathtt{sp}(H^0_B(H_j^k)) \subset \bigcup\limits_{i \in \Z}
      \mathtt{sp}(H^i_B(\mathcal{K}^k_{i+j})) =
      \bigcup\limits_{\substack{i \in \Z \\ i+j \leq k}}
      \bigcup\limits_{v \in \bm{\Sigma}^k_{i+j}}
      \mathtt{sp}(H^i_B(\K[\bm{x}](-v)))$$
    }
      By \Cref{thm:regularityR},
    $\mathtt{sp}(H_B^i(\K[\bm{x}](-v))) \subset
    \bigcup\limits_{\substack{\alpha \in \{1,\dots,r\} \\ N_\alpha + 1
        = i \\ \alpha \neq \emptyset}} Q_\alpha + v$.
    The proposition follows by a change of indices.
  \end{proof}

  \begin{proposition} \label{thm:localCohomOfKoszul}
    If $(f_1,\dots,f_k)$ form a regular sequence outside $B$, then
    for $i > 0$, $H_B^i(H^k_j) = 0$ and
    for $j > 0$, it holds $H_B^0(H^k_j) = H^k_j$.
  \end{proposition}

  The proposition follows from considering the spectral sequence of
  the double complex given by the Koszul complex and the \v{C}ech complex
  of $f_1,\dots,f_k$ over $B$, when $f_1,\dots,f_k$ is a regular
  sequence outside $B$, \cite[Sec.~4]{awane2005formes}.
  
  \begin{corollary}[Multihomogeneous Macaulay bound] \label{thm:multihomMacBound}
    Let $f_1,\dots,f_{N+1}$ be regular sequence outside $B$ and 
    $\bm{D_k} := \left( \sum_{i = 1}^k \deg(f_i) \right) - \bm{n}$.
    If $\bm{d} \geq {\bm{D_k}}$, then $\forall i,j, k$,
    $[H_B^j(H^k_i)]_{\bm{d}} = 0$.
  \end{corollary}

  \begin{proof}
    We use \Cref{thm:bound_reg}. Fix $i$ and $j$ in
    \Cref{eq:localCohom}, and consider
    $\alpha \subset \{1, \dots, k\}$ such that $N_\alpha+1+j-i \leq k$,
    $\#\alpha \neq \emptyset$, and
    $v \in \bm{\Sigma}^k_{N_\alpha+1+j-i}$. If $t \in \alpha$, then the
    $t$-th coordinate of any element in $Q_\alpha + v$ has to be
    $\leq -n_t - 1 + v_t$, where $v_t$ is the $t$-th coordinate of
    $v$. As all the multidegrees $\deg(f_1),\dots,\deg(f_k)$ are
    non-negative, $v_t \leq \sum_{i=1}^k deg(f_i)_t$. So,
    $-n_t - 1 + v_t < -n_t + \sum_{i=1}^k deg(f_i)_t = ({\bm{D_k}})_t
    \leq \bm{d}_t$. Hence, $\bm{d} \not\in Q_\alpha + v$. By
    \Cref{thm:bound_reg}, $[H_B^0(H^k_i)]_{\bm{d}} = 0$.
  \end{proof}

  The bound ${\bm{D_k}}$ is not tight, e.g. see \cite[Sec.~4.4]{awane2005formes}.

  Like with homogeneous polynomials, we define the
  multigraded Hilbert function, $HF$,  of a $\K$-module $\Modu$ as the function
  that maps the multidegrees $\bm{d} \in \Z^r$ to
  $HF(\Modu,\bm{d}) = \dim_{\K}([\Modu]_{\bm{d}})$. When $\bm{d}$ is,
  component-wise, big enough, then $HF(\Modu,\bm{d})$ equals a
  polynomial $P_\Modu \in \Q[y_1,\dots,y_r]$ evaluated at $\bm{d}$
  \cite[Prop.~2.8]{maclagan2003uniform}; the
  Hilbert polynomial.
  If all the local cohomologies of $\Modu$ at a multidegree $\bm{d}$
  vanish,
  that is
  for all $i$, $[H_B^i(\Modu)]_{\bm{d}} = 0$,
  then, for this $\bm{d}$, the  Hilbert
  function and polynomial agree, $HF(\Modu,\bm{d}) = P_\Modu(\bm{d})$
  \cite[Prop.~2.14]{maclagan2003uniform}.
  
  \begin{corollary} \label{thm:dimensionMinusRank}
    Let $\bm{{d}} \geq D_K$, component-wise. If $k = N$, then the
    dimension of
    $[ \K[\bm{x}] / \langle f_1,\dots,f_N \rangle]_{\bm{{d}}}$ is
    the number of solutions, counting multiplicities, of the system 
    $(f_1,\dots,f_N)$ over $\P$.
    When $k = N + 1$,
    $\K[\bm{x}]_{\bm{{d}}} = [\langle f_1,\dots,f_{N+1}
    \rangle]_{\bm{{d}}}$.
  \end{corollary}
  
  \subsection{Computing graded parts of the ideals}
  \label{sec:0-dimensional-square}
  Let $(f_1,\dots,f_{k})$ be multihomogeneous system over $\P$.
  \Cref{alg:matrixMultihom} computes a set of generators of the vector
  space $[\langle f_1,\dots,f_k \rangle]_{\bm{d}}$.  Moreover, if
  $(f_1,\dots,f_k)$ form a regular sequence outside $B$, and
  ${\bm{d}} \geq {\bm{D_k}}$, then it performs no reduction to zero.

  \begin{algorithm}[]
     \caption{$\texttt{M}_3\texttt{H}(\{f_1,\dots,f_{k}\}, \bm{d}, <)$}
     \begin{algorithmic}
     \label{alg:matrixMultihom}
     \REQUIRE $f_1,\dots,f_k \in \K[\bm{x}]$, degree
       $\bm{d}$ and $<$ a monomial order
       %

     \STATE $\mathfrak{L}$ $\leftarrow$ $\emptyset$.

     \IF{$k = 1$}
         \STATE $\M^k_{\bm{d}} \leftarrow$ \parbox[t]{300px}{Macaulay matrix
           with columns indexed by the monomials in $\K[\bm{x}]_{\bm{d}}$
           in decreasing order wrt~$<$}
         
     \ELSE

     \STATE $\M^k_{\bm{d}} \leftarrow$
     \parbox[t]{175px}{$\texttt{M}_3\texttt{H}(\{f_1,\dots,f_{k-1}\}, \bm{d}, <)$}

     \STATE
     $\mathfrak{L}$ $\leftarrow$ \parbox[t]{300px}{Leading
       monomials of the Gaussian elimination of
       $\texttt{M}_3\texttt{H}(\{f_1,\dots,f_{k-1}\}, \bm{d}-\deg(f_k), <)$}

     \ENDIF

     \FOR{$\bm{x}^\beta \in \K[\bm{x}]_{\bm{d}-\deg(f_k)}$}
        \IF{$\bm{x}^\beta \not\in \mathfrak{L}$}
           \STATE Add to $\M^k_{\bm{d}}$ the polynomial $\bm{x}^\beta \cdot f_k$
        \ENDIF
     \ENDFOR
        
     \RETURN $\M^k_{\bm{d}}$
     \end{algorithmic}
   \end{algorithm}

  \begin{theorem} \label{thm:algMultiHomComputesGenerators}
    Let $(f_1,\dots,f_{k})$ be a multihomogeneous
    system. \Cref{alg:matrixMultihom} computes a matrix such that the
    polynomials in its rows form a set of generators of the vector
    space $[\langle f_1,\dots,f_k \rangle]_{\bm{d}}$,
    $\forall {\bm{d}} \in \Z^r$.
  \end{theorem}

  We omit the proof as it is similar to
  \Cref{thm:givenGBformMacMatrix,thm:sparse-f5}.

  \begin{remark} \label{thm:vanishingOfH1}
    Following the definition of the Koszul complex, \linebreak
    $[H_i^k]_{\bm{d}}=0$ implies that, given any syzygy 
    $\sum_i g_i \cdot f_i = 0$ such that
    $deg(g_i f_i) = \bm{d}$, then $\forall j$,
    $g_j \in [\langle f_1,\dots,f_{j-1},f_{j+1},\dots,f_k
    \rangle]_{\bm{d} - \deg(f_j)}.$
  \end{remark}
  
  \begin{lemma} \label{thm:h1impliesLastFullRank}
    If $[H_1^k]_{\bm{d}} = 0$, then every polynomial $\bm{x}^\beta\!\cdot\!f_k$
    in $\M^k_{\bm{d}}$ is linear independent to the (polynomials corresponding to) other rows.
  \end{lemma}

  \begin{proof}
    If there is a polynomial of the form $\bm{x}^\beta \cdot f_k$ in
    $\M^{k}_{\bm{d}}$ that is linearly dependent with the other rows
    of the matrix, then there is a syzygy of the system
    $(f_1,\dots,f_k)$ involving $f_k$. That is, there are
    multihomogeneous polynomials $g_1,\dots,g_k$ so that
    $\sum_i g_i \, f_i = 0$, for every $\bm{x}^\sigma$ in the support
    of $g_i$ it holds $\bm{x}^\sigma \cdot f_i \in \rows(\M^{k}_{\bm{d}})$,
    and $\bm{x}^\beta$ belongs to the support of $g_k$. As $H_1^k$
    vanishes at degree ${\bm{d}}$, by \Cref{thm:vanishingOfH1},
    $g_k \in [\langle f_1,\dots,f_k
    \rangle]_{{\bm{d}}-\deg(f_k)}$. But, by construction,
    $LM(g_k) \cdot f_k$ does not belong to
    $\rows(\M^{k}_{\bm{d}})$. Hence, this syzygy can not be formed
    with the rows of $\M^k_{\bm{d}}$.
  \end{proof}
  
  \begin{lemma}
     If $[H_1^s]_d = 0$, for all $s \leq k$, then $M^k_d$ is full-rank.
  \end{lemma}

   \begin{proof}
     We proceed by induction on $k$. The case $k = 1$ is trivial, as
     $\langle f_1 \rangle$ is a principal ideal. If $M^{k}_d$ is not
     full-rank, we have a syzygy involving $f_k$, because $M^{k-1}_d$
     is full-rank by inductive hypothesis. Hence, there are
     multihomogeneous polynomials $g_1,\dots,g_k$ such that
     $\sum_i g_i \, f_i = 0$ and we can form each $g_i$ with the rows
     of $M^{k}_d$. As $H_1^k$ vanishes at degree $d$, then
     $g_k \in [\langle f_1,\dots,f_k \rangle]_{d-\deg(f_k)}$. Hence,
     the $LM(g_k) \cdots f_k$ does not belong to the $Rows(M^{k}_d)$,
     so we can not have this syzygy.
   \end{proof}
  
  \begin{corollary} \label{thm:noRedToZeroMultihom}
    If $(f_1,\dots,f_k)$ is a regular sequence outside $B$, then for
    ${\bm{d}} \geq {\bm{D_k}}$, all the matrices appearing in
    \Cref{alg:matrixMultihom} are full-rank.
  \end{corollary}

  \begin{proof}
    We proceed by induction on $k$. When $k=1$, the ideal is principal
    and so the theorem holds.
    In step $k$, note that ${\bm{d}} \geq {\bm{D_k}}$ implies
    ${\bm{d}} \geq {\bm{d}} - deg(f_k) \geq {\bm{D_k}} - \deg(f_k) =
    \bm{D_{k-1}}$. Hence, we have no reduction to zero in the
    recursive calls.
    As ${\bm{d}} \geq {\bm{D_k}}$, by \Cref{thm:localCohomOfKoszul},
    $H_B^0(H_i^k) = H_i^k$, and by \Cref{thm:multihomMacBound},
    $[H_i^k]_{\bm{d}} = 0$. Hence, by
    \Cref{thm:h1impliesLastFullRank}, $\M^k_{\bm{d}}$ has not
    reduction to zero involving $\bm{x}^\beta \cdot f_k$. As, by
    induction, $\texttt{M}_3\texttt{H}(\{f_1,\dots,f_{k-1}\}, {\bm{d}}, <)$
    is full-rank, $\M^k_{\bm{d}}$ is full-rank.
  \end{proof}

  \subsection{Solving zero-dimensional systems}
  Our solving strategy is to dehomogenize the system and to compute
  the multiplication maps for the affine variables.  Then we can apply
  FGLM to compute a Gr\"{o}bner basis or  compute the
  eigenvalues/eigenvectors of the multiplication maps.

  Let $(f_1,\dots,f_N)$ be a $0$-dimensional system over $\P$ with no
  solutions at infinity.
  If we do not know if the system has no solutions at infinity,
  we can ensure it by
  performing a generic linear change of coordinates preserving the
  multihomogeneous structure, e.g. see \cite[Pg.~121]{cox2006using}.
  We use \Cref{alg:matrixMultihom} to construct a monomial basis and
  the multiplication maps over
  $\K[\bar{\bm{x}}] / \langle \bar{f_1},\dots,\bar{f_N} \rangle$.
  Following \Cref{alg:matrixMultihom}, let
  $\mathfrak{L}$ be the set of leading monomials of the polynomials
  in $[\langle f_1,\dots,f_N \rangle]_{{\bm{D_N}}}$, with respect to $<$.
  Let $\mathfrak{b}$ be a list of monomials in
  $\K[\bm{x}]_{{\bm{D_k}}}$ not in $\mathfrak{L}$, sorted by $<$.
  Consider
  $\bm{D_{N+1}} := \bm{{\bm{D_N}}} + \bm{\bar{1}}$.

  \begin{definition}
    For a multilinear polynomial $f_0 \in \K[\bm{x}]_{\bm{\bar{1}}}$,
    let $\widecheck{\M}^{f_0}$ be the Macaulay matrix that we obtain
    after we permute the columns of \linebreak
    $\texttt{M}_3\texttt{H}(\{f_1,\dots,f_{N},f_{0}\}, \bm{D_{N+1}}, <)$ so
    that the columns indexed by the monomials
    $\{ \bm{x}_h \cdot \bm{x}^\beta : \bm{x}^\beta \in \b \}$ are the
    last ones.
    Let $\widecheck{\M}^{f_0}$ be 
    $\bigl[\begin{smallmatrix} M_{1,1}^{f_0} & M_{1,2}^{f_0} \\
      M_{2,1}^{f_0} & M_{2,2}^{f_0} \end{smallmatrix} \bigr]$, where
    the monomials indexing the columns of
    $\bigl[\begin{smallmatrix} M_{1,2}^{f_0} \\
      M_{2,2}^{f_0} \end{smallmatrix} \bigr]$ are the monomials in
    $\{ \bm{x}_h \cdot \bm{x}^\beta : \bm{x}^\beta \in \b \}$, and the
    polynomials in the rows of
    $\bigl[\begin{smallmatrix} M_{2,1}^{f_0} &
      M_{2,2}^{f_0} \end{smallmatrix} \bigr]$ are of the form
    $\{\bm{x}^\beta \cdot f_0 : \bm{x}^\beta \in \b\}$.
  \end{definition}

   Observe that, the matrix
    $\bigl[\begin{smallmatrix} M_{1,1}^{f_0} &
      M_{1,2}^{f_0} \end{smallmatrix} \bigr]$ is a permutation of  \linebreak
    $\texttt{M}_3\texttt{H}(\{f_1,\dots,f_{N} \}, \bm{D_{N+1}}, <)$, and the
    polynomials in its rows do not involve $f_0$, so we can forget the
    superscripts.
  
  \begin{remark} \label{thm:noSolMfullrank} By
    \Cref{thm:dimensionMinusRank}, if $(f_1,\dots,f_N)$ is
    $0$-dimensional, and $f_0$ does not vanish on
    $V_\P(f_1,\dots,f_N)$, then $\widecheck{\M}^{f_0}$ is invertible.
  \end{remark}
  
  \begin{theorem} \label{thm:monomialBasis}
    Let $\bar{\b}$ be the dehomogenization of the monomials in $\b$.
    If the system $f_1,\dots,f_n$ has no solutions at infinity, then
    $\bar{\b}$ forms a monomial basis for
    $\K[\bar{\bm{x}}] / \langle \bar{f_1}, \dots, \bar{f_N} \rangle$.
  \end{theorem}

  \begin{proof}
    The set $\bar{\b}$ is a monomial basis if its elements are linear
    independent on
    $\K[\bar{\bm{x}}] / \langle \bar{f_1}, \dots, \bar{f_N} \rangle$
    and generate this quotient ring. By \Cref{thm:dimensionMinusRank},
    the dimension of the quotient ring, as a vector space, is the same
    as the number of elements in $\bar{\b}$, so we only need to prove
    the linear independence of the elements in $\bar{\b}$.
    Assume that there is a linear combination
    $\bar{p} := \sum_i c_i \bar{\b}_i$ 
    congruent to $0$ in
    $\K[\bar{\bm{x}}] / \langle \bar{f_1}, \dots, \bar{f_N} \rangle$.
    Then, similarly to \Cref{rmk:biggerDegSameDehom}, there is a
    $\omega \in \N$, such that
    $(\bm{x}_h)^\omega \cdot p \in \langle f_1, \dots, f_N \rangle$,
    where $p := \sum_i c_i \b_i$.
%
    By \Cref{thm:noSolMfullrank}, as the system has no solutions at
    infinity, $\widecheck{\M}^{\bm{x}_h}$ is invertible. The rows of
    $\widecheck{\M}^{\bm{x}_h}$ contain the set
    $\{\bm{x}_h \cdot \b_i\}_i$, so
%
    we can form $\bm{x}_h \cdot p$ by taking a linear combination of
    them. As the matrix is full-rank, this row is independent from the
    polynomials in $[\langle f_1,\dots,f_N \rangle]_{\bm{D_{N+1}}}$
    (\Cref{thm:algMultiHomComputesGenerators}), and then
    $\bm{x}_h \cdot p \not\in [\langle f_1,\dots,f_N
    \rangle]_{\bm{D_{N+1}}}$. Hence, $\omega > 1$.
    The multidegree of $(\bm{x}_h)^\omega \cdot p$ is
    $\bm{D_N} + \omega~\cdot~\bm{\bar{1}}$. As $\omega > 1$, 
    ${\bm{D_N}} + \omega~\cdot~\bm{\bar{1}} \geq \bm{D_{N+1}}$. By
    \Cref{thm:multihomMacBound},
    $[H_1(\mathcal{K}_\bullet(f_1,\dots,f_N, \bm{x}_h \, ; \,
    \K[\bm{x}]))]_{\bm{D_N} + \omega~\cdot~\bm{\bar{1}}} = 0$. Then, by
    \Cref{thm:vanishingOfH1},
    $(\bm{x}_h)^{\omega-1} \cdot p \in \langle f_1, \dots, f_N
    \rangle$. But, assuming minimality of $\omega$,
    $(\bm{x}_h)^{\omega-1} \cdot p \not\in \langle f_1, \dots, f_N
    \rangle$. So, $\bar{p}$ does not exist.
  \end{proof}

  \begin{remark}
    If the system $(f_1,\dots,f_N, \bm{x}_h)$ has no solutions over
    $\P$, by \Cref{thm:noSolMfullrank}, the matrix $M^{\bm{x}_h}$ is
    invertible. As $M_{2,1}^{\bm{x}_h}$ is zero, and $M_{2,2}^{\bm{x}_h}$
    is the identity, the matrix
    $M^{\bm{x}_h}$ is invertible.
  \end{remark}


  \begin{definition}
    When {\small$(f_1\dots f_N)$} has no solutions at infinity,~we
    define {\small $(M_{2,2}^{f_0})^c := M_{2,2}^{f_0} - M_{2,1}^{f_0} \cdot
    M_{1,1}^{-1} \cdot M_{1,2}$}, the {\small\emph{Schur complement}}
    of~{\small $M_{2,2}^{f_0}$}.
  \end{definition}

  \begin{theorem}
    If the system $(f_1,\dots,f_N)$ has no solutions at infinity, then
    the matrix $(M_{2,2}^{f_0})^c$ is the multiplication map of
    $\bar{f_0}$ over
    $\K[\bar{\bm{x}}] / \langle \bar{f_1}, \dots, \bar{f_N} \rangle$,
    with respect to the basis $\bar{\b}$.
  \end{theorem}

  \begin{proof}
    By \Cref{thm:monomialBasis}, $\bar{\b}$ is a monomial
    basis of
    $\K[\bar{\bm{x}}] / \langle \bar{f_1}, \dots, \bar{f_N}
    \rangle$. Hence, for every $i$, 
    $\b_i \cdot f_0 \equiv \bm{x}_h \sum_j (M_{2,2}^{f_0})^c_{i,j}
    \b_j \mod \langle f_1,\dots,f_N \rangle$. If we dehomogenize,
    $\bar{\b}_i \cdot \bar{f_0} \equiv \sum_j (M_{2,2}^{f_0})^c_{i,j}
    \bar{\b}_j \mod \langle \bar{f_1},\dots,\bar{f}_N \rangle$.
  \end{proof}


  \paragraph*{Acknowledgments:}
    We thank Laurent Bus\'e, Marc Chardin, and Joaqu\'in Rodrigues
    Jacinto for the helpful discussions and references.
    We thanks the
    anonymous reviewers for their detailed comments and suggestions.
    The authors are partially supported by ANR JCJC GALOP
    (ANR-17-CE40-0009) and the PGMO grant GAMMA.

  \bibliographystyle{alpha}
  \bibliography{sparsegb}

\end{document}